\documentclass[12pt]{article}
\usepackage{amsmath,amssymb,amsthm,mathrsfs}
\usepackage{array}
\usepackage{graphicx,psfrag,epsf}
\usepackage{multirow}
\usepackage{multicol}
\usepackage{multirow}
\usepackage{booktabs}
\usepackage{dsfont}
\usepackage{appendix}
\usepackage{float}
\usepackage[colorlinks,citecolor=blue,linkcolor=blue,urlcolor=blue]{hyperref}
\usepackage{algorithm,algpseudocode}
\usepackage[bottom=0.5in,top=1.5in]{geometry}

\newcommand{\blind}{1}
\newtheorem{lemma}{Lemma}
\newtheorem{definition}{Definition}

\newtheorem{remark}{Remark}
\DeclareMathOperator*{\argmax}{\arg\!\max}

\begin{document}

\def\spacingset#1{\renewcommand{\baselinestretch}%
{#1}\small\normalsize} \spacingset{0.95}


\if1\blind
{
  \title{\bf Estimating
    heterogeneous treatment effects versus 
    building individualized treatment rules: Connection and disconnection}
  \author{Zhongyuan Chen and Jun Xie \\ \\
   {\it Department of Statistics, Purdue University} \\
     {\it 150 N. University Street, West Lafayette, IN 47907} \\
   junxie@purdue.edu 
    }
  \date{\today}
  \maketitle
} \fi

\if0\blind
{
  \bigskip
  \bigskip
  \bigskip
  \begin{center}
    {\LARGE\bf Estimating heterogeneous treatment effects versus
      building individualized treatment rules} 
\end{center}
  \medskip
} \fi

\spacingset{1.45} 

\begin{abstract}
Estimating heterogeneous treatment effects is a well studied topic in
the statistics literature. More recently, it has regained
attention due to an increasing need for precision medicine as well as
the increased use of state-of-art machine learning methods in the
estimation. Furthermore, estimating heterogeneous treatment effects is
directly related to building an individualized treatment rule, which is
a decision rule of treatment according to patient
characteristics. This paper examines the connection and disconnection
between these two research problems. Notably, a better estimation of the
heterogeneous treatment effects may or may not lead to a better
individualized treatment rule. We provide 
theoretical frameworks to explain the connection and
disconnection and demonstrate two different scenarios through
simulations. Our conclusion sheds light on a practical guide that 
under certain circumstances, there is no need to enhance estimation of
the treatment effects, as it does not alter the treatment 
decision.

\end{abstract}


\vfill\eject

\section{Introduction}
Estimation of heterogeneous treatment effects (HTE) is a research problem
commonly raised in many fields, including politics, economics, education,
and healthcare. Instead of an overall treatment effect, heterogeneity
exists for subgroups or individuals within a population. We
represent HTE by the conditional average treatment effect (CATE)
function given a set of covariates. There is a rich literature on
estimating the CATE, mostly through regression methods, as well as
more recently developed machine learning algorithms
\cite{Kunzel19}. In theory, the accuracy of a CATE estimator is
measured by the expected mean squared error (EMSE), which 
is based on the quadratic loss function. The convergence rate of the EMSE
can be used to compare different CATE estimators.

Estimating heterogeneous treatment effects is critical for
making medical decisions, such as what treatment to recommend. Building
personalized 
treatment, also referred to as an individualized treatment rule (ITR)
\cite{Qian11}, is usually done through estimating the CATE first and
then defining the optimal ITR as the sign function of the
CATE (for two treatment options denoted as 1 and $-1$). \cite{Chen22}
provided a good introduction on building ITRs and called 
the aforementioned method an indirect approach. Although related,
estimating the CATE and building optimal ITRs are separately studied
in the literature. In this paper, we elaborate upon the connection and
disconnection between the two research areas with mathematical
frameworks. 

In Section 2, we review recent developments on the estimation of the CATE,
including the use of a machine learning
method named X-learner \cite{Kunzel19}, which can improve estimation
accuracy with a faster rate of convergence than other
conventional estimators. In Section 3, we
connect estimating the CATE to building an ITR and show that the
optimal ITR is indeed the sign function of the true CATE function. With this
relationship, it is natural to
expect that better estimators, in terms of smaller EMSE, lead to
improved ITRs. This, however, does not always happen due to the mismatch
between the loss functions of these two research problems. More
specifically, we use the quadratic loss function for estimation of the
CATE but the 0-1 loss function for comparing ITRs. In Section 4, we
examine the disconnection through 
a mathematical framework. That is, for many cases, improving the CATE
estimators does not change the corresponding ITRs. Our conclusion sheds
light on a practical guide that in certain situations, there is no
need to enhance estimation of the CATE, as it does not alter the
treatment decision. Section 5 provides simulation examples to display
the connection and disconnection.

\section{Recent development on estimation of HTE}
Suppose we have data from a
two-arm clinical trial with $(X, A, Y)$, where $Y \in \mathbb{R}$ denotes
a treatment response variable (the larger value the better), $X \in
{\mathbb X} \subset \mathbb{R}^p$ is a set of covariates, and $A \in
{\cal A}=\{-1,1\}$ denotes the treatment index corresponding to the
control or treatment arm. We assume $(X, A, Y) \sim {\cal P}$, where
${\cal P}$ is the distribution from a specific family.
Denote the conditional mean $E(Y|X, A)$. Define
\[\mu_0(x)=E(Y|X=x, A=-1) \hspace{5mm} \text{and} \hspace{5mm}
  \mu_1(x)=E(Y|X=x, A=1).
\]
Then, the CATE function is $\tau(x) = \mu_1(x)-\mu_0(x)$. Let
$\hat{\tau}(x)$ be an estimator of the CATE from a set of independent
random data from ${\cal P}$. We are interested in
estimators with a small expected mean squared error (EMSE):
\[
  \text{EMSE}({\cal P}, \hat{\tau})= E[(\hat{\tau}({\cal X})-\tau({\cal X}))^2],
\]
where the expectation is taken over $\hat{\tau}$ and ${\cal X}$, which are
assumed independent of each other, e.g., $\hat{\tau}$ is estimated
from a training data set and ${\cal X}$ denotes a new data.
 
Various methods are available to estimate the CATE. The most commonly
used one is to fit regression models for $\hat{\mu}_1(x)$ and $\hat{\mu}_0(x)$,
and then $\hat{\tau}(x) =\hat{\mu}_1(x)-\hat{\mu}_0(x)$. As an
addition to the rich literature, \cite{Chen22} considered very general
regression models and applied dimension reduction for high-dimensional
covariates. Moreover,
supervised learning algorithms are used to estimate $\hat{\mu}_1$
and $\hat{\mu}_0$ by the machine learning community \cite{Hu21},
including Bayesian Additive Regression Trees (BART) \cite{Chipman10}
and Random Forest (RF) \cite{Wager18}. One of the most recent
developments is an algorithm called X-learner \cite{Kunzel19}, which
specifically exploits structural properties of the
CATE function for an improved estimator.

Following the results of \cite{Kunzel19}, we use the convergence rates
of the EMSE to compare different CATE 
estimators. Suppose we observe independent and identically distributed data
$(X_i, A_i, Y_i) \sim {\cal P}$, $i=1,\ldots, N$, with $m$ control
units and $n$ treated units, $N=m+n$. Let $n$ denote the smaller
sample size of the two 
treatment arms and assume $m$ and $n$ have a similar scale. Most of the
estimation methods have a convergence rate 
that depends on the estimators of $\hat{\mu}_1$ and
$\hat{\mu}_0$. For instance, for a parametric distribution family,
e.g., $\mu_0$ and $\mu_1$ are parametric functions and $X$ and $Y$
follow parametric distributions, the ordinary least squares estimators
of $\hat{\mu}_1$ and $\hat{\mu}_0$ achieve the optimal EMSE minimax
rates of $\mathcal{O}(n^{-1})$. Therefore, the corresponding
$\hat{\tau}= \hat{\mu}_1-\hat{\mu}_0$ also has the EMSE rate of
$\mathcal{O}(n^{-1})$. In general, for data from a specific distribution family, 
${\cal P} \in S$, we denote the EMSE rate of 
$\hat{\tau}$ as $\mathcal{O}(n^{-a_{\mu}})$, where $0<a_{\mu}\le 1$ is the
EMSE rate of $\hat{\mu}_1$ and $\hat{\mu}_0$.

Through imputation of individual treatment effects, \cite{Kunzel19}
developed a new estimation method of the CATE function, namely the
X-learner. The X-learner uses the observed data to impute the
unobserved individual treatment effects and then directly estimate the
CATE function. Denote the X-learner's convergence rate as
$\mathcal{O}(n^{-a_{\tau}})$. When the CATE function has a simpler
structure than $\mu_0$ and $\mu_1$, e.g., zero or approximately
linear, or the number of observations in one treatment group (usually
the control group) is much larger than that in the other, it is proven that
$a_{\tau}>a_{\mu}$. To conclude, estimators of the CATE can often
be improved in terms of a faster convergence rate of the EMSE. There is
a clear tendency to choose the estimator with the smaller
EMSE.  

\section{Connection between the two research areas}
We have represented HTE by the CATE
function, $\tau(x)=E(Y|X=x, A=1)-E(Y|X=x, A=-1)$. We now introduce
ITR, which is a function $d(x): {\mathbb X} \to {\cal
  A}$. That is, an ITR is a map from the space of covariates, e.g.,
prognostic variables, to the space of treatments. An optimal ITR,
denoted as $d_0(x)$, is the function that gives 
the highest mean response.  

We use the ITR framework from \cite{Qian11, Chen22}. Recall $(X, A, Y)
\sim {\cal P}$. For a given ITR $d$, let ${\cal P}^d$ denote the 
distribution of $(X, A, Y)$ given that $d(X)$ is used to assign
treatments. Then, the expected treatment response under the ITR $d$ is 
\begin{equation} \label{V}
    \int Y d{\cal P}^d = \int Y \frac{d{\cal P}^d}{d{\cal P}} d{\cal P}
    =  \int Y \frac{\mathds{1}_{d(X)=A}}{p(A|X)}d{\cal P} = E\left[Y \frac{\mathds{1}_{d(X)=A}}{p(A|X)}\right],
\end{equation}
where $p(A|X)$ is the treatment assignment probability and is assumed
$>0$. This expectation is called the value function and denoted by 
$V(d)$. Formally, the optimal ITR $d_0$ is the rule that maximizes $V(d)$,
that is, $d_0 \in \argmax_{d} V(d)$.

We can obtain the optimal ITR $d_0$ by estimating the CATE. To see
this, we denote $Q_0(X,A)\triangleq E(Y|X, A)$ and $Q_0(x,a)$ is
sometimes referred to as the treatment response function. The CATE
function is then
\[
  \tau(x)=Q_0(x, 1) - Q_0(x, -1).
\]
  
\begin{lemma}
  Assume $p(a|x)>0$ for $a=1, -1$ and $x \in {\mathbb X}$. The
  optimal ITR $d_0$ is the sign function 
   $d_0(x)={\text{sign}} (\tau(x))$, for $x$ such that $\tau(x) \ne
  0$.
\end{lemma}

\begin{proof}
Note that upper case letters are used to denote random variables and
lower case letters for values of the random variables. Using Formula
(\ref{V}), we have 
\[
    V(d) = E\left[\frac{\mathds{1}_{d(X)=A}}{p(A|X)} E[Y|X,A]\right] = E\left[ \sum\limits_{a\in \mathcal{A}}\mathds{1}_{d(X)=a}Q_0(X,a)\right] = E\left[Q_0(X,  d(X))\right].
\]
The value for the optimal ITR $V(d_0) = E[Q_0(X,
  d_0(X))] \leq E[\max\limits_{a\in\mathcal{A}} Q_0(X,
  a)]$. Meanwhile by the definition of $d_0$,
\[V(d_0)
\geq V(d)|_{d(X) \in \argmax_{a\in\mathcal{A}} Q_0(X, a)} =
E[\max\limits_{a\in \mathcal{A} }Q_0(X, a)].
\]
Thus, $V(d_0) = E[\max\limits_{a\in \mathcal{A} }Q_0(X, a)] $ and the
optimal ITR satisfies $d_0(X)
=\argmax\limits_{a\in\mathcal{A}}Q_0(X,a)$. In other words, we have 
\[
d_0(x) = \left\{\begin{array}{ll}
1, & \text{if } Q_0(x, 1)>Q_0(x,-1),\\
-1, & \text{if } Q_0(x, 1)<Q_0(x,-1).
\end{array}
\right.
\]
That is,
\begin{equation} \label{sign}
  d_0(x)={\text{sign}} (\tau(x))
\end{equation}
for $x$ that $\tau(x) \ne 0$.
\end{proof}  

Lemma 1 indicates that building an optimal ITR can be achieved by
accurate estimation of the CATE. We can further justify this by
providing a relationship between the value function and the estimation
error. More specifically, \cite{Chen22} have showed that for
any ITR $d$,  the reduction in value is upper bounded by the
estimation error (See Lemma 1 in \cite{Chen22}):
\begin{equation} \label{upbound}
 V(d_0) - V(d) \leq C\left(E[(Q(X,A)-Q_0(X,A))^2]\right)^{1/2} ,
\end{equation}
where $Q(X,A)$ is a function $Q:\mathbb{X} \times \mathcal{A} \to
\mathbb{R}$ such that $d(X)=\argmax_{a\in 
  \mathcal{A}} Q(X, a)$. If we consider $Q(x,a)$ as an estimator of
$Q_0(x,a)$, then $\hat{\tau}(x) \triangleq Q(x, 1) - Q(x, -1)$ is the
estimated CATE function and we also have the corresponding 
ITR $d(x)={\text{sign}} (\hat{\tau}(x))$.

Intuitively, this upper bound implies that if the EMSE of $Q$ is
small, then the corresponding estimated ITR $d$ is closer to the
optimal ITR $d_0$ in terms of the value function. Formula
(\ref{upbound}) along with Lemma 1 explain the
connection between estimating the CATE and building an ITR. They
support the approach of minimizing the estimation error of the CATE
function and then setting the ITR as the sign function of the
estimator, i.e., $d(x)={\text{sign}} (\hat{\tau}(x))$. 

\section{Disconnection}
Despite the relationship, it is interesting to notice that in Formula
(\ref{upbound}), minimizing
the upper bound is quite different from minimizing the original value
function difference, $V(d_0) - V(d)$, as the upper bound may not be very
tight. In fact, we will formally demonstrate that improving the estimation of
the CATE function (with smaller EMSE) does not necessarily improve
ITR.

Besides the value function, we use a misclassification error to
evaluate ITRs. Formula (\ref{sign}) shows the optimal ITR $d_0$ is
${\text{sign}}(\tau(x))$, where $\tau(x)$ is the
true CATE function. For any ITR $d(x)$, which is also a binary
decision rule, we define the expected misclassification rate as
follows: 
\begin{equation}\label{error}
R(d)=P(d(X) \ne d_0(X)), 
\end{equation}
where $P$ is the marginal distribution of $X \in {\mathbb X}$ when
$(X, A, Y) \sim {\cal P}$.   
This error is based on the 0-1 loss, while the EMSE
is based on the quadratic loss. This mismatch is displayed in Figure
\ref{fig:loss}. For simplicity, imagine we have a 
true positive treatment effect located at 1 on the horizontal
axis. Any estimators of the treatment effects located on the positive side of
the horizontal axis give correct treatment decisions. That is, the sign
functions agree with the sign of the true treatment
effect. Analogously, any estimators on the negative  
side of the axis give incorrect treatment decisions based on the
sign functions. Consider two estimators of the treatment effects and
the corresponding ITRs from the sign functions. The expected misclassification
rates of the two ITRs  will not change if they fall on the same side
around 0, even though the estimators give very different quadratic
losses. 

\begin{figure}[ptbh] 
\centering
\includegraphics[width=0.5\linewidth]{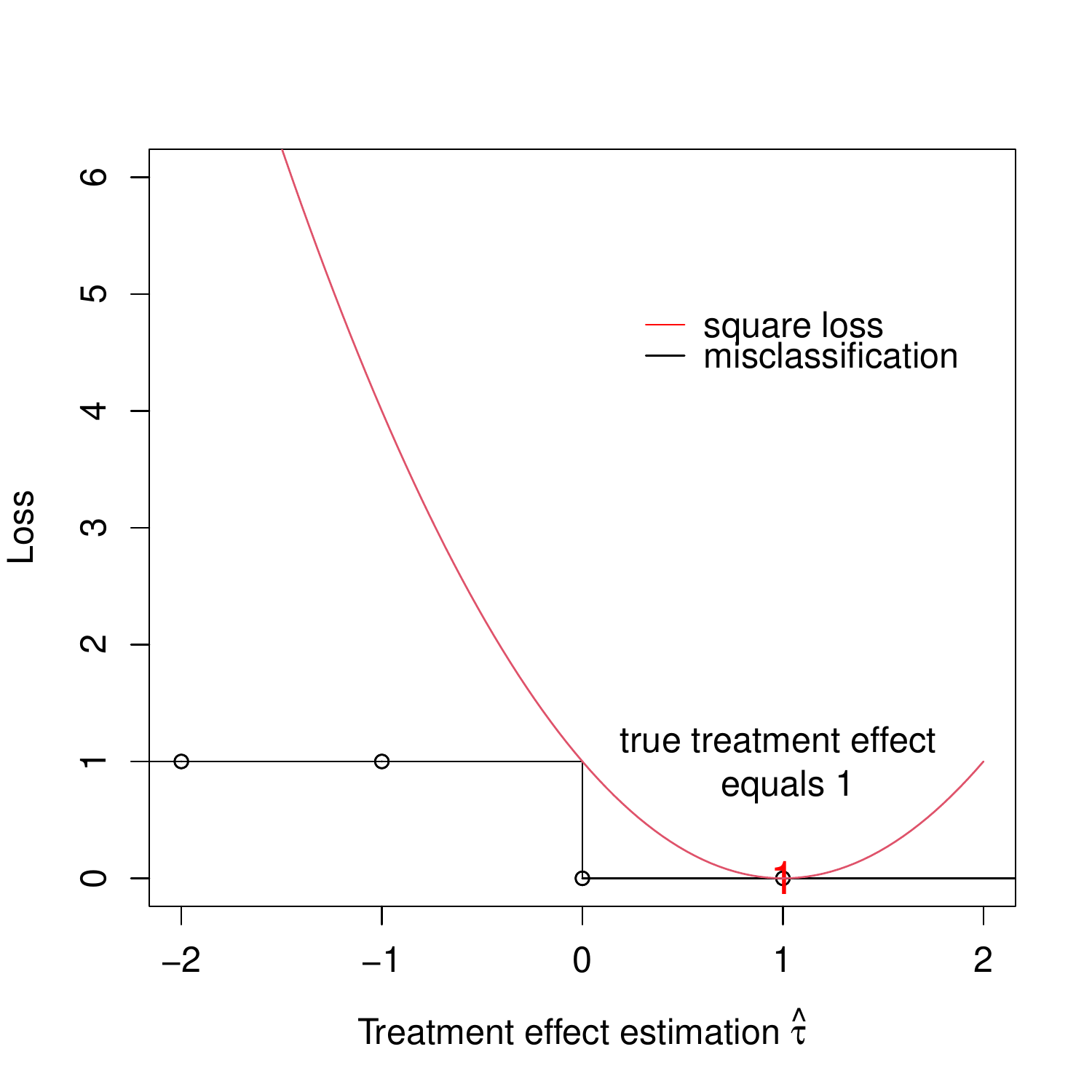}
\caption{\small{Mismatch of the two loss functions: square loss and 0-1 loss}}\label{fig:loss}
\end{figure}

Given the data
$(X_i, A_i, Y_i) \sim {\cal P}$, $i=1,\ldots, N$, with $m$ control
units and $n$ treated units, $N=m+n$,  we
consider two estimators of the CATE, $\hat{\tau}_1$ and
$\hat{\tau}_2$. For instance, $\hat{\tau}_1$ denotes the X-learner and
$\hat{\tau}_2= \hat{\mu}_1 - \hat{\mu}_0$ is from the ordinary least
squares estimators $\hat{\mu}_1$ and $\hat{\mu}_0$ as discussed in
Section 2. There are two corresponding ITRs as the sign functions, 
denoted as $\hat{d}_1$ and $\hat{d}_2$, respectively. We now evaluate
the CATE estimators and the ITRs by comparing them in terms of the
EMSE and the expected misclassification rate, respectively.

Essentially, for any given estimation method, its performance depends 
on the specific data distribution $(X_i, A_i, Y_i) \sim {\cal P}$. We
can talk about a property of ${\cal P}$ that will be held for a family
of the distributions ${\cal P}$. More specifically,
given two estimation methods $\hat{\tau}_1$ and
$\hat{\tau}_2$ as denoted above, we can define a family of distributions
as follows, 
\[
  S=\{{\cal P} \text{ such that EMSE}({\cal P}, \hat{\tau}_1) <
  \text{EMSE}({\cal P}, \hat{\tau}_2) \}.
\]
For example, for distribution
families that satisfy Conditions 1-6 as stated in \cite{Kunzel19}, we
know the EMSE of the X-learner is smaller than the other standard learner
methods as long as the sample sizes $m$ and $n$ are large enough (Theorem 1-2 in
\cite{Kunzel19}). In other words, the relationship,  $\text{EMSE}({\cal P},
\hat{\tau}_1) < \text{EMSE}({\cal P}, \hat{\tau}_2)$, is valid for a whole
distribution family, ${\cal P} \in S$, where ${\cal P}$ is the
underlying data-generating distribution.
We further define different distribution families that satisfy different
relationships between $\hat{\tau}_1$ and $\hat{\tau}_2$.
\begin{definition} \label{def:families}
Consider $(X_i, A_i, Y_i) \sim {\cal P}$, $i=1,\ldots, N$, with $m$
control units and $n$ treated units and $N=m+n$, $0<m$, $n<N$. For two
estimation strategies of the CATE function, 
$\hat{\tau}_1$ and $\hat{\tau}_2$, we define three distribution
families satisfying the following conditions:
\begin{enumerate}
\item[(a)] Let $S_0$ denote the set of all distributions ${\cal P}$ such that
  $\text{EMSE}({\cal P}, \hat{\tau}_1) < \text{EMSE}({\cal P},
  \hat{\tau}_2)$.
\item[(b)] Let $S_1$ denote the set of all distributions ${\cal P}$ such
  that $(\hat{\tau}_1(x)-\tau(x))^2 < (\hat{\tau}_2(x)-\tau(x))^2$
for every $x \in \mathbb X$. That is, the pointwise squared error
of $\hat{\tau}_1$ is smaller than $\hat{\tau}_2$.
\item[(c)] Let $S_2$ denote the set of all distributions ${\cal P}$ such
  that $R(\hat{d}_1) < R(\hat{d}_2)$. That is, the expected misclassification
  rate of $\hat{d}_1$ is smaller than $\hat{d}_2$, where
  $\hat{d}_1(x) ={\text{sign}}(\hat{\tau}_1(x))$ and $\hat{d}_2(x) ={\text{sign}}(\hat{\tau}_2(x))$.
\end{enumerate}
\end{definition}

\begin{lemma} \label{lemma2}
 Consider two estimation strategies for the CATE function, e.g., the
 X-learner and the standard ordinary least squares estimator denoted as
$\hat{\tau}_1$ and $\hat{\tau}_2$. Following Definition
\ref{def:families}, we have
\[
  S_2 \subset S_1 \subset S_0, 
\]
where $\subset$ denotes a strict subset.
\end{lemma}

The proof is provided in the Appendix.

\begin{remark}
  Lemma \ref{lemma2} indicates that there are many situations where
improving a CATE estimator does not change ITRs. In fact, when 
$\hat{\tau}_1(x)$ and $\hat{\tau}_2(x)$ have the same sign for $x \in
\mathbb X$, the corresponding ITRs are exactly the same with no
improvement. 
\end{remark} 
  
\begin{remark}
Lemma \ref{lemma2} suggests a practice guideline when we apply
different estimation methods for the CATE function. 
We really need to examine the sign functions of different
estimators. When there is no change of signs between two estimators,
or the changes are 
minimal, we should choose the estimation method based on the
computational cost instead of the EMSE accuracy. 
\end{remark}  

Note that the
X-learner is always computationally more expensive than other
estimation algorithms, because the X-learner requires several additional
computation steps of data imputation. We should not
use the X-learner when there is minimal difference between
the sign functions of the X-learner and that of the other estimation
methods.

\section{Simulations}
We study two simulation scenarios as specified in Table
\ref{tab:scenario}. The first scenario demonstrates a case when a better
CATE estimation leads to a better ITR. The second 
scenario, on the other hand, demonstrates when a better estimation does not
lead to a better ITR. We generate samples for the treatment and
control arms independently. We simulate a covariate $X_i
\sim \text{Unif}[-1, 1]$ and then the response 
$Y_i=\mu(X_i)+\epsilon_i$, where $\mu(x)=\mu_1(x)$ 
for the treatment arm, and $\mu(x)=\mu_0(x)$ for the
control arm, and
$\epsilon_i$ is the random error following a normal distribution
with mean zero and standard deviation 0.01. In the training data, we
vary sample sizes 
of the control group 
and the treatment group,  $m= 200, 400, 600, 800$ and $n=10, 20, 30,
40$. In the testing data, $m=1000$ and $n=50$. Our simulations mimic
the situations where the number of
samples in the control group is much larger than that in the treatment group.
\begin{table}[ptbh]
\caption{\small{Two scenario settings in the simulation studies. } } 
\label{tab:scenario}
\vskip5mm\centering\tabcolsep4mm
\begin{tabular}{cc|cc}\hline
\multicolumn{4}{c}{Scenario 1}                   \\\hline	
treatment $A=1$	 &  &	control $A=-1$&	\\\hline
range of x	& response function &	range of x &	response function\\\hline
$-1< x <-0.4$&	$\mu_1(x)=1.1$&	$-1< x <-0.4$&	$\mu_0(x)=1.0$\\\hline
$-0.4< x< 0.5$&	$\mu_1(x)=1.6$&	$-0.4< x< 0.5$&	$\mu_0(x)=1.5$\\\hline
$0.5< x <1.0$&	$\mu_1(x)=1.1$&	$0.5< x <1.0$&	$\mu_0(x)=1.0$\\\hline
\multicolumn{4}{c}{Scenario 2}                   \\\hline	
treatment   $A=1$&  &	control $A=-1$	&\\\hline
range of x& 	response function&	range of x&	response function\\\hline
$-1< x <-0.4$&	$\mu_1(x)=1.6$&	$-1< x <-0.4$&	$\mu_0(x)=1.0$\\\hline
$-0.4< x< 0.5$&	$\mu_1(x)=2.1$&	$-0.4< x< 0.5$&	$\mu_0(x)=1.5$\\\hline
$0.5< x <1.0$&	$\mu_1(x)=1.6$&	$0.5< x <1.0$&	$\mu_0(x)=1.0$\\\hline
\end{tabular}
\end{table}

We consider two machine learning methods for estimating the CATE: the
first estimator is the X-learner, denoted as $\hat{\tau}_X$.
The second estimator, denoted as $\hat{\tau}$, is a
standard estimator $\hat{\tau}= \hat{\mu}_1 - \hat{\mu}_0$. For
both estimators, we use the simple linear regression to estimate the
response function $\mu_1$ in the treatment group, as the sample size is small.
We, however, use a locally weighted regression method, LOESS
\cite{Cleveland88}, to estimate $\mu_0$.
We compare the two methods across four aspects: 1) EMSE: empirical mean squared error calculated as
$\sum_{i=1}^{N}(\hat{\tau}(X_i)-\tau(X_i))^2/N$; 2)  misclassification
rate: the percentage of the number of misclassified treatments, where
the misclassified treatment means 
the recommended treatment disagrees with the truth; 3) empirical
value function: an estimate of the value function as in
\cite{Qian11,Zhao12}; and 4) computational time.  
 
The plots on the left column in Figure \ref{fig:sim} (A,C,E,G)
demonstrate the 
connection between treatment effect estimation and ITR. In this 
scenario, a better treatment effect estimation results in a better
ITR. In particular, the
performances of $\hat{\tau}_X$ are better than $\hat{\tau}$
in terms of smaller EMSE, lower
misclassfication rate, and higher value function. The improvements
are bigger with larger sample sizes. On the other
hand, the plots on the right column in Figure \ref{fig:sim} (B,D,F,H)
illustrate the disconnection between treatment effect estimation and
ITR. That is, although the 
EMSE with $\hat{\tau}_X$ is smaller than that with $\hat{\tau}$, the
misclassification rate and the value function are almost
the same using $\hat{\tau}_X$ and $\hat{\tau}$. In other words, a
better treatment effect estimation does not result in a better ITR in
this scenario. 

\begin{figure}[ptbh]
\centering
\includegraphics[width=0.9\linewidth]{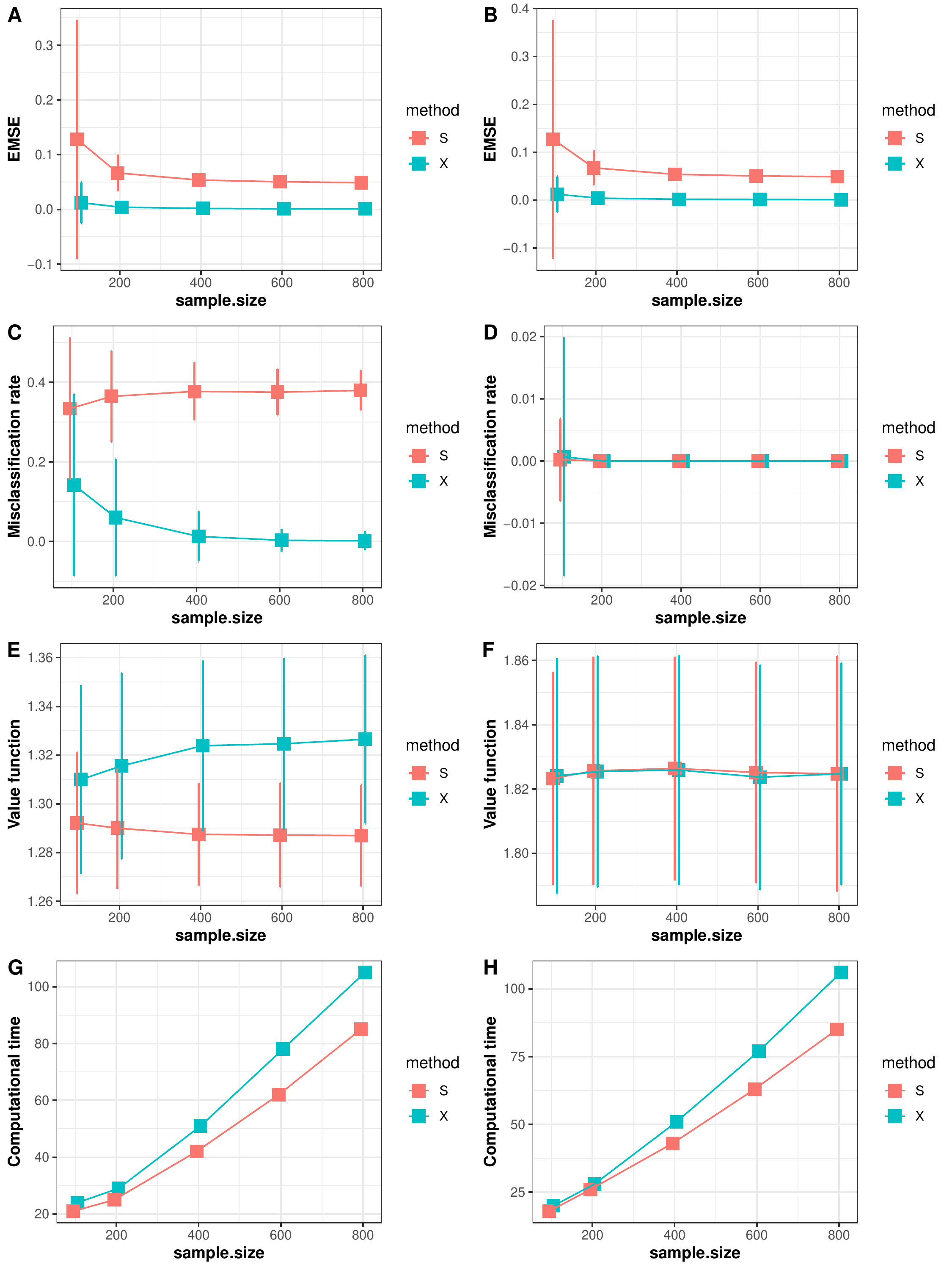}
\vskip-4mm 
\caption{\small{Comparison of estimator $\hat{\tau}_X$ (X) from the X-learner and standard estimator $\hat{\tau}$ (S) in terms of 
: A.B. Empirical mean square error (EMSE) of treatment estimator; C.D. Misclassification rates; E.F. Empirical value function; G.H. Computational time (unit: seconds) from 1000 simulations. left column: scenario 1; right column: scenario 2.}}\label{fig:sim}
\end{figure}

It is important to notice that the computational time of
$\hat{\tau}_X$ is longer than that of $\hat{\tau}$ for both scenarios.
In the second scenario, since $\hat{\tau}_X$ does not improve ITR, we should
not use $\hat{\tau}_X$ but prefer the standard estimation method for
ITR, with more efficient computation.

\section{Discussion}
The literature on estimation of heterogeneous treatment effects mainly uses
the EMSE for evaluations of estimation methods. If our ultimate goal
is to make treatment recommendations, we show that the evaluation
results from the HTE literature are not adequately useful. This
conclusion may be surprising but will enable us to properly choose
an estimation method for obtaining ITRs.

\section*{Appendix: Proof of Lemma 2}
\begin{proof}
Following Definition \ref{def:families}(a) and \ref{def:families}(b),
and the EMSE defintion, 
$\text{EMSE}({\cal P}, \hat{\tau})= E[(\hat{\tau}({\cal X})-\tau({\cal
  X}))^2]$, it 
is straightforward that for any ${\cal P} \in S_1$ we must have ${\cal P} 
\in S_0$. Therefore, we  
obtain $S_1 \subset S_0$ and $S_1$ as a strict subset of $S_0$.

Next, for an ITR denoted as $d(x)$ and any $x \in \mathbb X$, define
\[
R(d(x)) = \left\{\begin{array}{ll}
1, & \text{if } d(x) \ne d_0(x),\\
0, & \text{if } d(x) = d_0(x),
\end{array}
\right.   
\]
where $d_0$ is the optimal ITR and $d_0(x)={\text{sign}}
(\tau(x))$. That is, $R(d(x))$ is the pointwise misclassification
error of the ITR $d$ at $x \in \mathbb X$. 

For ${\cal P} \in S_1$, from Definition \ref{def:families}(b), we know
$(\hat{\tau}_1(x)-\tau(x))^2 < (\hat{\tau}_2(x)-\tau(x))^2$ for any
$x \in \mathbb X$. Without loss of generality, assume
$\tau(x)>0$. There are three different cases for 
the misclassification error of $\hat{d}_1(x)$ and $\hat{d}_2(x)$:
\begin{enumerate}
\item $R(\hat{d}_1(x)) = R(\hat{d}_2(x))$, if $\hat{\tau}_1(x)$ and
  $\hat{\tau}_2(x)$ are on the same side of the origin, either both
  positive or both negative, and $|\hat{\tau}_1(x)-\tau(x)| <
  |\hat{\tau}_2(x)-\tau(x)|$. 
\item $R(\hat{d}_1(x)) > R(\hat{d}_2(x))$, if
  $|\hat{\tau}_1(x)-\tau(x)| < |\hat{\tau}_2(x)-\tau(x)|$ but
  $\hat{\tau}_1(x)<0<\tau(x)<\hat{\tau}_2(x)$.
\item $R(\hat{d}_1(x)) < R(\hat{d}_2(x))$, if
  $\hat{\tau}_2(x)<0<\hat{\tau}_1(x)<\tau(x)$, or
  $\hat{\tau}_2(x)<0<\tau(x)<\hat{\tau}_1(x)$ and
  $|\hat{\tau}_1(x)-\tau(x)| < |\hat{\tau}_2(x)-\tau(x)|$. 
\end{enumerate}
Integrating $R(\hat{d}_1(x))$ and 
$R(\hat{d}_2(x))$ over $x \in \mathbb X$ with respect of the marginal
distribution of $X$, we obtain the misclassification errors for the two 
ITRs $\hat{d}_1$ and $\hat{d}_2$, respectively. The distribution family 
$S_2 = \{{\cal P} \text{ such that } R(\hat{d}_1)<R(\hat{d}_2) \}$ 
corresponds to Item 
3 from the list. Therefore, for any ${\cal P} \in S_2$ we must have ${\cal 
P} \in S_1$. We obtain $S_2 \subset S_1$ and $S_2$ as a strict subset.

\end{proof}


\bibliographystyle{plainnat}

\begin{thebibliography}{}
\itemsep1pt

\bibitem[Bai et al. 2017]{Bai17} Bai X, Tsiatis AA, Lu W, Song R,
  (2017). Optimal treatment regimes for survival endpoints using a
  locally-efficient doubly-robust estimator from a classification
  perspective. {\em Lifetime Data Anal} 23, 585-604.

\bibitem[Chen et al. 2022]{Chen22} Chen Z, Wang Z, Song Q, and Xie J,
  (2022). Data-guided Treatment Recommendation with Feature
  Scores. {\em Statistica Sinica} 32, 2497-2519. 

\bibitem[Chipman et al. 2010]{Chipman10} Chipman HA, George EI, and
  Mcculloch RE, (2010). BART: Bayesian Additive Regression Trees. {\em
    Annals of Applied 
Statistics} 4(1), 266–298.
  
\bibitem[Cleveland and Devlin 1988]{Cleveland88} Cleveland W.S. and
  Devlin S.J. (1988). Locally Weighted Regression: An Approach to
  Regression Analysis by Local Fitting. {\em Journal of the American
    Statistical Association} 83, 596-610. 

\bibitem[Hu et al. 2021]{Hu21} Hu L, Ji J, and Li F
  (2021). Estimating heterogeneous survival treatment effect in
  observational data using machine learning. {\em Statistics in
    Medicine}, 40, 4691–4713.  
  
\bibitem[Kunzel 2019]{Kunzel19} Kunzel SR, Sekhona JS,
  Bickel PJ, and Yu B, (2019). Metalearners for estimating
  heterogeneous treatment 
  effects using machine learning. {\em PNAS} 116, 4156-4165.

\bibitem[Lin 2004]{Lin04} Lin Y, (2004). A note on margin-based loss
  functions in classification. {\em Statist. \& Prob. Letters} 68, 73-82.

  
\bibitem[Qian and Murphy 2011]{Qian11} Qian M. and Murphy
    S.A. (2011). Performance guarantees for individualized treatment
  rules. {\em The Annals of Statistics} 39, 1180-1210.


\bibitem[Wager and Athey 2018]{Wager18} Wager S and Athey S (2018)
  Estimation and Inference of Heterogeneous Treatment Effects using
  Random Forests, {\em Journal of the American Statistical 
Association}, 113, 1228-1242.


\bibitem[Zhao et al. 2012]{Zhao12} Zhao Y., Zeng D., Rush A.J.,
    Kosorok M.R.  (2012). Estimating Individualized Treatment Rules Using Outcome Weighted Learning. {\em Journal of the American
Statistical Association} 107, 1106–1118.

\bibitem[Zhou et al. 2017]{Zhou17} Zhou X., Mayer-Hamblett N., Khan U., Kosorok M.R. (2017).
Residual Weighted Learning for Estimating Individualized Treatment Rules. {\em Journal of the American
Statistical Association} 112, 169-187.
  
\end{thebibliography}

\end{document}